\documentclass[12pt,draftcls,onecolumn]{IEEEtran}

\usepackage{amssymb}
\usepackage{amsmath}
\usepackage{url}
\usepackage{graphicx}

\newtheorem{lemma}{Lemma}
\newtheorem{corollary}{Corollary}

\author{Frank Nielsen,~\IEEEmembership{Senior Member,~IEEE}
and Richard Nock,~\IEEEmembership{Nonmember}
\thanks{Frank Nielsen is with Sony Computer Science Laboratories, Inc., 3-14-13 Higashi Gotanda, 141-0022 Shinagawa-ku, Tokyo, Japan, {\tt nielsen@csl.sony.co.jp}}%
\thanks{Richard Nock is with UAG CEREGMIA, Martinique, France, {\tt rnock@martinique.univ-ag.fr}.}
}

\def\E{\mathcal{E}}
\def\X{\mathcal{X}}
 
\def\KL{\mathrm{KL}}
\def\dx{\mathrm{d}x}

\def\inner#1#2{\langle {#1}, {#2}\rangle}

\def\dnu{\mathrm{d}\nu}

\def\calX{\mathcal{X}}

\def\calE{\mathcal{E}}

\def\dxi{\mathrm{d}\xi}

\def\dim{\mathrm{dim}}

\def\dnu{\mathrm{d}\nu}
\def\Poi{\mathrm{Poi}}
\def\Nor{\mathrm{Nor}}

\title{On the Chi square and  higher-order Chi distances for approximating $f$-divergences}

\begin{document}
\maketitle

\IEEEcompsoctitleabstractindextext{%
\begin{abstract}
We report  closed-form formula for calculating the Chi square and higher-order Chi  distances between statistical distributions belonging to the same exponential family with affine natural space, and instantiate those formula for the Poisson and isotropic Gaussian families.
We then describe an analytic formula for the $f$-divergences based on Taylor expansions and relying on an extended class of Chi-type distances.  
\end{abstract}

}

\maketitle 

\IEEEdisplaynotcompsoctitleabstractindextext
\IEEEpeerreviewmaketitle

\section{Introduction}

\subsection{Statistical divergences: $f$-divergences}

Measuring the similarity or {\em dissimilarity} between two probability measures is met ubiquitously in signal processing.
Some usual distances  are the Pearson $\chi^2_P$  and Neyman $\chi^2_N$ chi square distances, and the Kullback-Leibler divergence~\cite{ct-1991} defined respectively by:
\begin{eqnarray}
\chi^2_P(X_1:X_2) &=& \int \frac{(x_2(x)-x_1(x))^2}{x_1(x)} \dnu(x),\\ 
\chi^2_N(X_1:X_2) &=& \int \frac{(x_1(x)-x_2(x))^2}{x_2(x)} \dnu(x),\\
\KL(X_1:X_2) &=& \int x_1(x) \log \frac{x_1(x)}{x_2(x)} \dnu(x),
\end{eqnarray}
where $X_1$ and $X_2$ are probability measures absolutely continuous with respect to a reference measure $\nu$, and $x_1$ and $x_2$ denote their  Radon-Nikodym densities, respectively.
Those dissimilarity measures $M$ are termed {\em divergences} to contrast with metric distances since they are oriented distances ({\it i.e.}, $M(X_1:X_2)\not = M(X_2:X_1)$) that do not satisfy the triangular inequality.
In the 1960's, many of those divergences were unified using the generic framework of $f$-divergences~\cite{Amari-2009,Csiszar-2013}, $I_f$, defined for an arbitrary functional $f$:
 
\def\dom{\mathrm{dom}}
\begin{equation}\label{eq:fdiv}
I_f(X_1:X_2) = \int x_1(x) f\left(\frac{x_2(x)}{x_1(x)}\right) \dnu(x) \geq 0,
\end{equation} 
where $f$ is a convex function $f: (0,\infty)\subseteq \dom(f) \mapsto [0,\infty]$ such that $f(1)=0$).
Indeed, it follows from Jensen inequality that $I_f(X_1:X_2)\geq f(\int x_2(x)\dnu(x))=f(1)=0$.
Furthermore,  wlog., we may consider $f'(1)=0$ and fix the scale of divergence by setting $f''(1)=1$, see~\cite{Amari-2009}.
Those $f$-divergences\footnote{Beware that sometimes  the $\chi^2_N$ and $\chi^2_P$ definitions are inverted in the literature. This may stem from an alternative definition of $f$-divergences defined as 
$I_f'(X_1:X_2) = \int x_2(x) f(\frac{x_1(x)}{x_2(x)}) \dnu(x)=I_f(X_2:X_1)$.} can always be symmetrized by taking
$S_f(X_1:X_2)= I_f(X_1:X_2) + I_{f^*}(X_1:X_2)$,
 with $f^*(u)= u f(1/u)$, and $I_{f^*}(X_1:X_2)=I_f(X_2:X_1)$.
See Table~\ref{tab:fdiv} for a list of common $f$-divergences with their corresponding generators $f$.
In information theory, $f$-divergences are characterized as the {\em unique} family of convex separable~\cite{Amari-2009} divergences that satisfies the {\em information monotonicity} property~\cite{informationgeometry-2000}.



Note that $f$-divergences may evaluate to infinity (that is, {\em unbounded} $I_f$) when the integral diverge, even if $x_1,x_2>0$ on the support $\mathcal{X}$.
For example, let $\X=(0,1)$ be the unit interval, and two densities (with respect to Lebesgue measure $\nu_L$) $x_1(x)=1$ and $x_2(x)=c e^{-1/x}$ with $c^{-1}=\int_0^1 e^{-1/x}\dx \simeq 0.148$ the normalizing constant. Consider the Kullback-Leibler divergence ($f$-divergence with $f(u)=u\log u$):
$\KL(X_1:X_2) =  \int_0^1 x_1 \log \frac{x_1(x)}{x_2(x)} \dnu_L(x)=-\log c + \int_0^1 \frac{1}{x}\dnu(x) = \infty$.

\begin{table*}
$$
\begin{array}{lll}
\text{Name of the $f$-divergence } & \text{Formula $I_f(P:Q)$} & \text {Generator $f(u)$ with $f(1)=0$}\\
\hline\hline
\text{Total variation (metric)} & \frac{1}{2}\int |p(x)-q(x)| \dnu(x) & \frac{1}{2} |u-1| \\
\text{Squared Hellinger} & \int (\sqrt{p(x)}-\sqrt{q(x)})^2 \dnu(x) & (\sqrt{u}-1)^2\\
\text{Pearson $\chi^2_P$}  &  \int \frac{(q(x)-p(x))^2}{p(x)} \dnu(x) & (u-1)^2\\
\text{Neyman $\chi^2_N$}  &  \int \frac{(p(x)-q(x))^2}{q(x)} \dnu(x) & \frac{(1-u)^2}{u}\\
\text{Pearson-Vajda $\chi^k_{P}$}  &   \int \frac{(q(x)-\lambda p(x))^k}{p^{k-1}(x)} \dnu(x)  & (u-1)^k\\
\text{Pearson-Vajda $|\chi|^k_{P}$}  &   \int \frac{|q(x)-\lambda p(x)|^{k}}{p^{k-1}(x)} \dnu(x)  & |u-1|^k\\
%
%
\text{Kullback-Leibler} & \int p(x)\log \frac{p(x)}{q(x)} \dnu(x) & -\log u\\
\text{reverse Kullback-Leibler} & \int q(x)\log \frac{q(x)}{p(x)} \dnu(x) & u\log u \\
\text{$\alpha$-divergence} &  \frac{4}{1-\alpha^2} (1-\int p^{\frac{1-\alpha}{2}}(x) q^{1+\alpha}(x) \dnu(x))  & \frac{4}{1-\alpha^2}(1-u^{\frac{1+\alpha}{2}})\\ 
\text{Jensen-Shannon} & \frac{1}{2}\int (p(x)\log \frac{2p(x)}{p(x)+q(x)} +  q(x)\log \frac{2q(x)}{p(x)+q(x)})\dnu(x) &  -(u+1)\log \frac{1+u}{2} + u\log u\\
\end{array}
$$
 
\caption{Some common $f$-divergences $I_f$ with corresponding generators: Except the total variation, $f$-divergences are not metric~\cite{tvfuniquemetric-2007}. \label{tab:fdiv}}
\end{table*}
 


\subsection{Stochastic  approximations of $f$-divergences}
To bypass the integral evaluation of $I_f$ of Eq.~\ref{eq:fdiv} (often mathematically intractable),  we carry out a stochastic integration:
\begin{equation}\label{eq:sto}
\widehat{I_f}(X_1:X_2) \sim \frac{1}{2n} \sum_{i=1}^n \left(f\left(\frac{x_2(s_i)}{x_1(s_i)}\right) +\frac{x_1(t_i)}{x_2(t_i)} f\left(\frac{x_2(t_i)}{x_1(t_i)}\right)\right),\end{equation}
with $s_1, ..., s_n$ and $t_1, ..., t_n$ IID. sampled from $X_1$ and $X_2$, respectively.
Those approximations, although converging to the true values when $n\rightarrow\infty$, are time consuming and yield poor results in practice, specially when the
dimension of the observation space, $\mathcal{X}$, is large. 
We therefore concentrate on obtaining exact or arbitrarily fine approximation formula for $f$-divergences by considering a restricted class of exponential families.

\subsection{Exponential families}

Let $\inner{x}{y}$ denote the inner product for $x,y\in\calX$: The inner product for vector spaces $\calX$ is the scalar product $\inner{x}{y}=x^\top y$.
An exponential family~\cite{BrownExpFam:1986}  is a set of probability measures $\E_F=\{P_\theta\}_\theta$ dominated by a measure $\nu$ having their Radon-Nikodym densities $p_\theta$ expressed canonically as:

\begin{equation}\label{eq:cf}
p_\theta(x) = \exp (\inner{t(x)}{\theta}-F(\theta)+k(x) ),
\end{equation}

\noindent for $\theta$ belonging to the {\em natural parameter space}:   
$\Theta = \left\{\theta\in \mathbb{R}^D \middle| \int  p_\theta(x) \dnu(x) =1 \right\}$.
Since $\log \int_{x\in\calX} p_\theta(x) \dnu(x)=\log 1=0$, it follows that
$F(\theta)=-\log \int  \exp (\inner{t(x)}{\theta}+k(x) )\dnu(x)$.
For full regular families~\cite{BrownExpFam:1986}, it can be proved that function $F$ is strictly convex and differentiable over the open convex set $\Theta$. Function $F$ characterizes the family, and bears different names in the literature (partition function,  log-normalizer or   cumulant function)
and parameter $\theta$  (natural parameter)  defines the member $P_\theta$ of the family $\E_F$.
Let $D=\dim(\Theta)$ denote the dimension of $\Theta$, the order of the family.
The map $k(x): \calX \rightarrow \mathbb{R}$ is an auxiliary function defining a carrier measure $\xi$ with $\dxi(x)=e^{k(x)} \dnu(x)$.
In practice, we often consider  the Lebesgue measure $\nu_L$  defined over the Borel $\sigma$-algebra $\calE=B(\mathbb{R}^d)$ of $\mathbb{R}^d$ for continuous distributions (e.g., Gaussian), or
the counting measure $\nu_c$ defined  on the power set $\sigma$-algebra $\calE=2^\calX$ for discrete distributions (e.g., Poisson or multinomial families).
The term $t(x)$ is a measure mapping called the sufficient statistic~\cite{BrownExpFam:1986}.
Table~\ref{tab:affspace} shows the canonical decomposition for the Poisson and isotropic Gaussian families.
Notice that   the Kullback-Leibler divergence between members $X_1\sim \E_F(\theta_1)$ and $X_2\sim \E_F(\theta_2)$ of the same exponential family amount to compute a Bregman divergence on swapped natural parameters~\cite{2011-brbhat}: $\KL(X_1:X_2)=B_F(\theta_2:\theta_1)$, where $B_F(\theta:\theta')=F(\theta)-F(\theta')-(\theta-\theta')^\top \nabla F(\theta')$, where $\nabla F$ denotes the gradient. 
 
\section{$\chi^2$ and higher-order $\chi^k$  distances}
%

\subsection{A closed-form formula}
When $X_1$ and $X_2$ belong to the same restricted exponential family $\E_F$, we obtain the following result:

\begin{lemma}\label{lemma1}
The Pearson/Neyman Chi square distance between $X_1\sim \E_F(\theta_1)$ and $X_2\sim \E_F(\theta_2)$ is given by:
\begin{eqnarray}
\chi_P^2(X_1:X_2) &=& e^{F(2\theta_2-\theta_1)- (2F(\theta_2)-F(\theta_1))}-1,\\
\chi_N^2(X_1:X_2) &=& e^{F(2\theta_1-\theta_2)- (2F(\theta_1)-F(\theta_2))}-1,
\end{eqnarray}
provided that $2\theta_2-\theta_1$ and $2\theta_1-\theta_2$ belongs to the natural parameter space $\Theta$.
\end{lemma}

This implies that the chi square distances are always bounded. 
The proof relies on the following lemma:

\begin{lemma}\label{lemma:int}
The integral $I_{p,q}=\int  x_1(x)^p x_2(x)^q \dnu(x)$ with $p+q=1$ for $X_1\sim \E_F(\theta_1)$ and $X_2\sim \E_F(\theta_2), p\in\mathbb{R}, p+q=1$ converge and equals to: 
\begin{equation}
I_{p,q}= e^{F(p\theta_1+q\theta_2)-(pF(\theta_1)+qF(\theta_2))}
\end{equation}
provided the natural parameter space $\Theta$ is {\em affine}.
\end{lemma}

\begin{proof}
Let us calculate the integral $I_{p,q}$:
\begin{eqnarray*}
&=& \int \exp (p(\inner{t(x)}{\theta_1}-F(\theta_1)+k(x))) \\
&& \times \exp (q(\inner{t(x)}{\theta_2}-F(\theta_2)+k(x))) \dnu(x),\\
&=& \int e^{\inner{t(x)}{p\theta_1+q\theta_2}-(pF(\theta_1)+qF(\theta_2))+k(x)}  \dnu(x),\\
&=& e^{F(p\theta_1+q\theta_2) - (pF(\theta_1)+qF(\theta_2))} \int  p_F(x| p\theta_1+q\theta_2) \dnu(x).
\end{eqnarray*}
When $p\theta_1+q\theta_2\in\Theta$, we have $\int  p_F(x|p\theta_1+q\theta2) \dnu(x)=1$, hence the result.
\end{proof}

To prove Lemma~\ref{lemma1}, we rewrite
$\chi^2_P(X_1:X_2) =\int (\frac{x_2^2(x)}{x_1(x)}-2x_2(x)+x_1(x)) \dnu(x) = \left(\int  x_1(x)^{-1}x_2(x)^2\dnu(x)\right)-1$,
and apply Lemma~\ref{lemma:int} for $p=-1$ and $q=2$ (checking that $p+q=1$). The closed-form formula for the Neyman chi square follows from the fact that $\chi^2_N(X_1:X_2)=\chi^2_P(X_2:X_1)$.
Thus when the natural parameter space $\Theta$ is affine, the Pearson/Neyman Chi square distances and its symmetrization $\chi_P^2+\chi_N^2$ between members of the same exponential family are available in closed-form.
Examples of such families are the Poisson, binomial, multinomial, or isotropic Gaussian families to name a few.
Let us call those families: {\em affine exponential families} for short.
The canonical decomposition of usual affine exponential families are reported in Table~\ref{tab:affspace}.
Note that a formula for the $\alpha$-divergences between members of the same exponential family were reported in~\cite{2011-brbhat} for $\alpha\in[0,1]$:
In that case, $\alpha\theta_1+(1-\alpha)\theta_2$ always belong to the open convex natural space $\Theta$ (here, $p$ belongs to $\mathbb{R}$).

\begin{table}

\begin{eqnarray*}
\Poi(\lambda) &:& p(x|\lambda)=\frac{\lambda^x e^{-\lambda}}{x!}, \lambda>0, x\in\{0, 1, ...\}\\
\Nor_I(\mu) &:& p(x|\mu)= (2\pi)^{-\frac{d}{2}} e^{-\frac{1}{2}(x-\mu)^\top (x-\mu)}, \mu\in\mathbb{R}^d, x\in\mathbb{R}^d\\
\end{eqnarray*}

$$
\begin{array}{|l||l|l|l|l|l|l|}\hline
\mathrm{Family} & \theta & \Theta & F(\theta) & k(x) & t(x) & \nu \\ \hline\hline
\mathrm{Poisson} & \log\lambda & \mathbb{R} & e^\theta & -\log x! & x & \nu_c\\
\mathrm{Iso. Gaussian} & \mu & \mathbb{R}^d & \frac{1}{2}\theta^\top\theta & \frac{d}{2}\log 2\pi-\frac{1}{2}x^\top x & x & \nu_L\\ \hline
\end{array}
$$

\caption{Examples of exponential families with affine natural space $\Theta$.
$\nu_c$ denotes the counting measure and $\nu_L$ the Lebesgue measure.
\label{tab:affspace}}
\end{table}
 
\subsection{The Poisson and isotropic Gaussian cases} 
 
As reported in Table~\ref{tab:affspace}, those Poisson and isotropic Gaussian exponential families have affine natural parameter spaces $\Theta$.

\begin{itemize}
\item The Poisson family.
For $P_1\sim \Poi(\lambda_1)$ and $P_2\sim \Poi(\lambda_2)$, we have:

\begin{equation}\label{eq:poi1}
\chi_P^2(\lambda_1:\lambda_2)=\exp \left(\frac{\lambda_2^2}{\lambda_1}-2\lambda_2+\lambda_1 \right) -1.
\end{equation}

To illustrate this formula with a numerical example, consider $X_1\sim \Poi(1)$ and $X_2\sim \Poi(2)$. Then, it comes that
$\chi_P^2(P_1:P_2)=e-1 \simeq 1.718$.

\item The isotropic Normal family.
For $N_1\sim \Nor_I(\mu_1)$ and $N_2\sim \Nor_I(\mu_2)$, we have according to Table~\ref{tab:affspace}:
$\chi_P^2(\mu_1: \mu_2)= e^{\frac{1}{2}(2\mu_2-\mu_1)^\top(2\mu_2-\mu_1) -(\mu_2^\top\mu_2-\frac{1}{2}\mu_1^\top\mu_1)} -1$.
In that case the $\chi^2$ distance is symmetric:
\begin{equation}\label{eq:nor1}
\chi_P^2(\mu_1: \mu_2)=e^{(\mu_2-\mu_1)^\top (\mu_2-\mu_1)} -1 =\chi^2_N(\mu_1:\mu_2)
\end{equation}

\end{itemize}

\subsection{Extensions to higher-order Vajda $\chi^k$ divergences}
The higher-order Pearson-Vajda $\chi^k_P$ and $|\chi^k_P|$ distances~\cite{CsiszarTaylor-2002} are defined by:
\begin{eqnarray}
\chi^k_P(X_1:X_2) &=& \int \frac{(x_2(x)-x_1(x))^k}{x_1(x)^{k-1}} \dnu(x),\\
|\chi|^k_P(X_1:X_2) &=& \int \frac{|x_2(x)-x_1(x)|^k}{x_1(x)^{k-1}} \dnu(x),
\end{eqnarray}
 are $f$-divergences for the
generators $(u-1)^k$ and $|u-1|^k$ (with $|\chi|^k_P(X_1:X_2)\geq \chi^k_P(X_1:X_2)$).
When $k=1$, we have $\chi^1_P(X_1:X_2)=\int (x_1(x)-x_2(x)) \dnu(x)=0$ ({\it i.e.}, divergence is never discriminative),
and $|\chi_P^1|(X_1,X_2)$ is twice the total variation distance (the only metric $f$-divergence~\cite{tvfuniquemetric-2007}).
$\chi_P^0$ is the unit constant.
Observe that the $\chi_P^k$ ``distance'' may be negative for odd $k$ (signed distance), but not $|\chi|^k_P$.
We can compute the $\chi^k_P$ term explicitly by performing the binomial expansion:

\begin{lemma}
The (signed) $\chi^k_P$ distance between members  $X_1\sim \E_F(\theta_1)$ and $X_2\sim \E_F(\theta_2)$ of the same affine exponential family is ($k\in\mathbb{N}$) always bounded and equal to:
\begin{eqnarray}
\chi^k_P(X_1:X_2)= 
 \sum_{j=0}^k {(-1)}^{k-j} {k \choose j}  \frac{e^{F((1-j)\theta_1+j\theta_2)}}{e^{(1-j)F(\theta_1)+jF(\theta_2)}}.
\end{eqnarray}
\end{lemma}

\begin{proof}
\begin{eqnarray}
&&\chi^k_P(X_1:X_2) = \int \frac{(x_2(x)-x_1(x))^k}{x_1(x)^{k-1}} \dnu(x),\\
&=& \int \sum_{j=0}^k (-1)^{k-j} {k \choose j} \frac{x_1(x)^{k-j} x_2(x)^{j}}{x_1(x)^{k-1}} \dnu(x),\\
&=&  \sum_{j=0}^k (-1)^{k-j} {k \choose j} \int x_1(x)^{1-j} x_2(x)^{j} \dnu(x).
\end{eqnarray}
Then the proof follows from Lemma~2 that shows that $I_{1-j,j}(X_1:X_2) =\int x_1(x)^{1-j} x_2(x)^{j} \dnu(x) =\frac{e^{F((1-j)\theta_1+j\theta_2)}}{e^{(1-j)F(\theta_1)+jF(\theta_2)}}$.
\end{proof}

For Poisson/Normal distributions, we get:

\begin{eqnarray}
\chi^k_P(\lambda_1:\lambda_2)&=&\sum_{j=0}^k (-1)^{k-j} {k \choose j} e^{\lambda_1^{1-j} \lambda_2^{j} - ((1-j)\lambda_1 +  j\lambda_2)},\\
\chi^k_P(\mu_1:\mu_2)&=&\sum_{j=0}^k (-1)^{k-j} {k \choose j} e^{\frac{1}{2}j(j-1)(\mu_1-\mu_2)^\top (\mu_1-\mu_2)}.
\end{eqnarray}

Observe that for $\lambda_1=\lambda_2=\lambda$, we have $\chi^k_P(\lambda_1:\lambda_2)=\sum_{j=0}^k (-1)^{k-j} {k \choose j}  e^{\lambda-\lambda}=(1-1)^k=0$ when $k\in\mathbb{N}$,
as expected. The $\chi^k_P$ value is always bounded.
For sanity check, consider the binomial expansion for $k=2$, we have:
$\chi^2_P(\lambda_1:\lambda_2)={2\choose 0} e^{\lambda_1-\lambda_1}-{2 \choose 1} e^{\lambda_2-\lambda_2}+{2\choose 2} e^{\frac{\lambda_2^2}{\lambda_1}-2\lambda_2}=e^{\frac{\lambda_2^2}{\lambda_1}-2\lambda_2}-1$, in accordance with Eq.~\ref{eq:poi1}.
Consider a numerical example: Let $\lambda_1=0.6$ and $\lambda_2=0.3$, then $\chi_P^2\sim 0.16$, $\chi_P^3\sim -0.03$, $\chi_P^4\sim 0.04$,
$\chi_P^5\sim -0.02$, $\chi_P^6\sim 0.018$, $\chi_P^7\sim -0.013$, $\chi_P^8\sim 0.01$, $\chi_P^9\sim -0.0077$, $\chi_P^{10}\sim 0.006$, etc.
This numerical example illustrates the alternating sign of those $\chi^k$-type signed distances.

\section{$f$-divergences from Taylor series}

Recall that the $f$-divergence defined for a generator $f$ is $I_f(X_1 : X_2) = \int x_1(x) f\left(\frac{x_2(x)}{x_1(x)}\right) \dnu(x)$.
Assuming $f$ analytic, we use the Taylor expansion about a point $\lambda$: $f(x)=f(\lambda)+ f'(\lambda)(x-\lambda)+\frac{1}{2}f''(\lambda)(x-\lambda)^2+...=\sum_{i=0}^{\infty} \frac{1}{i!} f^{(i)}(\lambda) (x-\lambda)^i$, the power series  expansion of $f$, for $\lambda\in \mathrm{int}(\dom(f^{(i)})) \forall i\geq 0$. 
 
\begin{lemma}[extends Theorem~1 of~\cite{CsiszarTaylor-2002}] 
When bounded, the $f$-divergence $I_f$ can be expressed as the power series of higher order Chi-type distances:
\begin{eqnarray}
I_f(X_1 : X_2) &=& \int x_1(x) \sum_{i=0}^{\infty} \frac{1}{i!}  f^{(i)}(\lambda) \left(\frac{x_2(x)}{x_1(x)}-\lambda\right)^i \dnu(x),\nonumber\\
&\stackrel{*}{=}& \sum_{i=0}^{\infty} \frac{1}{i!}  f^{(i)}(\lambda)\ \chi^i_{\lambda,P}(X_1:X_2),\label{eq:ftaylor}
\end{eqnarray}
\end{lemma}
In the $*$ equality, we  swapped the integral and sum according to Fubini theorem since 
we assumed that $I_f<\infty$, and $\chi^i_{\lambda,P}(X_1:X_2)$ is a generalization of the $\chi^i_{P}$ defined by:
\begin{equation}
 \chi^i_{\lambda,P}(X_1:X_2)=  \int  \frac{(x_2(x)-\lambda x_1(x))^i}{x_1(x)^{i-1}}  \dnu(x).
\end{equation} 
and $\chi^0_{\lambda,P}(X_1:X_2)=1$ by convention.
Note that $\chi^i_{\lambda,P}\geq f(1)=(1-\lambda)^k$ is a  $f$-divergence for $f(u)=(u-\lambda)^k-(1-\lambda)^k$ (convex for even $k$).
Eq.~\ref{eq:ftaylor} yields a meaningful numerical approximation scheme by truncating the series to the first $s$ terms, provided that the Taylor remainder is bounded.

\begin{itemize}

\item Choosing $\lambda=1\in\mathrm{int}(\dom (f^{(i)}))$, we approximate the $f$-divergence   as follows (Theorem~1 of~\cite{CsiszarTaylor-2002}):

\begin{eqnarray}\label{eq:fapprox}
\lefteqn{|I_f(X_1:X_2)-\sum_{k=0}^s \frac{f^{(k)}(1)}{k!} \chi^k_P(X_1:X_2)|}\nonumber\\ && \leq \frac{1}{(s+1)!} \|f^{(s+1)}\|_\infty (M-m)^s,
\end{eqnarray}
where $\|f^{(s+1)}\|_\infty=\sup_{t\in[m,M]} |f^{(s+1)}(t)|$ and $m\leq  \frac{p}{q}\leq M$.
Notice that by assuming the ``fatness'' of $\frac{p}{q}$, we ensure that $I_f<\infty$.

\item Choosing $\lambda=0$ (whenever $0\in \mathrm{int}(\dom(f^{(i)}))$) and affine exponential families,  we get the $f$-divergence   in a much simpler   analytic expression:

\begin{eqnarray}
I_f(X_1 : X_2) &=& \sum_{i=0}^{\infty} \frac{f^{(i)}(0)}{i!}   I_{1-i,i}(\theta_1:\theta_2),\\
I_{1-i,i}(\theta_1:\theta_2) &=& \frac{e^{F(i\theta_2+(1-i)\theta_1)}}{e^{i F(\theta_2)+(1-i)F(\theta_1)}}.
\end{eqnarray}
\end{itemize}

\begin{lemma}
The bounded $f$-divergences between members of the same affine exponential family
 can be computed as an equivalent power series whenever $f$ is analytic.
\end{lemma}

\begin{corollary}
A second-order Taylor expansion yields $I_f(X_1:X_2)\sim f(1)+f'(1)\chi^1_N(X_1:X_2)+\frac{1}{2} f''(1)\chi^2_N(X_1:X_2)$.
Since $f(1)=0$ ($f$ can always be renormalized) and $\chi^1_N(X_1:X_2)=0$, it follows that 
\begin{equation}
I_f(X_1:X_2)\sim \frac{f''(1)}{2}\chi^2_N(X_1:X_2),
\end{equation} and reciprocally $\chi^2_N(X_1:X_2) \sim  \frac{2}{f''(1)} I_f(X_1:X_2)$ ($f''(1)>0$ follows from the strict convexity of the generator). When $f(u)=u\log u$, this yields the well-known approximation~\cite{ct-1991}:
\begin{equation}
\chi^2_P(X_1:X_2)\sim 2\ \KL(X_1:X_2).
\end{equation}
\end{corollary}
For affine exponential families, we then plug the closed-form formula of Lemma 1 to get a simple approximation formula of $I_f$.
For example, consider the Jensen-Shannon divergence (Table~\ref{tab:fdiv}) with $f''(u)=\frac{1}{u}-\frac{1}{u+1}$ and $f''(1)=\frac{1}{2}$.
It follows that $I_{\mathrm{JS}}(X_1:X_2)\sim \frac{1}{4}\chi^2_N(X_1:X_2)$. (For Poisson distributions $\lambda_1=5$ and $\lambda_2=5.1$, we get $1.15\%$ relative error.


\subsection{Example 1: $\chi^2$ revisited}
Let us start with a sanity check for the $\chi^2$ distance between Poisson distributions. 
The Pearson chi square distance is a $f$-divergence for $f(t)=t^2-1$ with $f'(t)=2t$ and $f''(t)=2$ and $f^{(i)}(t)=0$ for $i>2$.
Thus, with $f^{(0)}(0)=-1$, $f^{(1)}(0)=0$, $f^{(2)}(0)=2$, and $f^{(i)}(0)=0$ for $i>2$.
Recall that  
$I_{1-i,i}(\theta_1:\theta_2) = e^{F(i\theta_2+(1-i)\theta_1)- (i F(\theta_2)+(1-i)F(\theta_1)}=\exp (\lambda_2^i \lambda_1^{1-i} - i\lambda_2 -(1-i)\lambda_1)$.
Note that $I_{1-i,i}(\lambda,\lambda)=e^0=1$ for all $i$.
Thus we get:
$I_f(X_1 : X_2) = - I_{1,0} + I_{-1,2}$ with $I_{1,0} = e^{\lambda_1-\lambda_1}=1$
and $I_{-1,2} = e^{\frac{\lambda_2^2}{\lambda_1} -2\lambda_2+\lambda_1}$.
Thus, we obtain $I_f(X_1 : X_2)= -1 + e^{\frac{\lambda_2^2}{\lambda_1} -2\lambda_2+\lambda_1}$,
in accordance with Eq.~\ref{eq:poi1}.

\subsection{Example 2: Kullback-Leibler divergence}
By choosing $f(u)=-\log u$, we obtain the Kullback-Leibler divergence (see Table~\ref{tab:fdiv}).
We have $f^{(i)}(u)={(-1)}^{i} (i-1)! u^{-i}$, and hence $\frac{f^{(i)}(1)}{i!}=\frac{{(-1)}^{i}}{i}$, for $i\geq 1$ (with $f(1)=0$).
Since $\chi^1_{1,P}=0$, it follows that:
\begin{equation}\label{eq:KLTaylor}
\KL(X_1:X_2)=\sum_{j=2}^\infty \frac{(-1)^i}{i}\ \chi^j_{P}(X_1:X_2).
\end{equation}

Note that for the case of   KL divergence between members of the same exponential families, the divergence can be expressed in a simpler closed-form using a Bregman divergence~\cite{2011-brbhat} on the swapped natural parameters.
For example, consider Poisson distributions with $\lambda_1=0.6$ and  $\lambda2=0.3$, the Kullback-Leibler divergence computed from
the equivalent Bregman divergence yields $\KL\sim 0.1158$, the stochastic evaluation of Eq.~\ref{eq:sto} with $n=10^6$ yields $\widehat{KL}\sim 0.1156$
and the KL divergence obtained from the truncation of Eq. \ref{eq:KLTaylor} to the first $s$ terms yields the following sequence:
$0.0809 (s=2)$, 
$0.0910 (s=3)$, 
$0.1017(s=4)$, 
$0.1135 (s=10)$,
$0.1150 (s=15)$, etc.

%

%
%

\section{Concluding remarks}
We investigated the calculation of statistical $f$-divergences between members of the same exponential family with affine natural space.
We first reported a generic closed-form formula for the Pearson/Neyman $\chi^2$ and Vajda $\chi^k$-type distance, and instantiated that formula for the Poisson and the isotropic Gaussian affine exponential families.
We then considered the Taylor expansion of the generator $f$ at any given point $\lambda$ to deduce an analytic expression of $f$-divergences using Pearson-Vajda-type distances.
A second-order Taylor approximation yielded a fast estimation of $f$-divergences.
This framework shall find potential applications in signal processing and when designing inequality bounds between divergences.

A Java\texttrademark{} package that illustrates numerically the lemmata is provided at:
\url{www.informationgeometry.org/fDivergence/}

\section*{Acknowledgments}
The authors would like to thank Professor Shun-ichi Amari for giving us his feedback and pointing out minor
careless mistakes in an early draft.



\end{document}